\documentclass[11pt]{article}
\usepackage{url}
\usepackage{epsf}
\usepackage{graphicx}
\usepackage{amsfonts}
\usepackage{amssymb}
\usepackage{amsmath}
\usepackage{latexsym}
\usepackage{setspace, float}

\newcommand{\ABox}{
\raisebox{3pt}{\framebox[6pt]{\rule{6pt}{0pt}}}
}
\newenvironment{proof}{{\bf Proof:}}{\hfill\ABox}

\newtheorem{theorem}{{\bf Theorem}}
\newtheorem{corollary}[theorem]{Corollary}
\newtheorem{lemma}[theorem]{Lemma}

\newcommand{\e}{\varepsilon}

\newcommand{\id}{{\tt ID}}

\newcommand{\yy}{YY}
\newcommand{\yys}{{\sc YY-Spanner}}

\newcommand{\YY}{{Yao-Yao}}
\newcommand{\ys}{{\sc Yao-S{tep}}}
\newcommand{\rys}{{\sc Reverse Yao-S{tep}}}

\newcommand{\yss}{{\sc Sink-S{tep}}}

\newcommand{\yssa}{{\sc Yao-Sparse-Sink}}
\newcommand{\Y}{{Yao}}

{\makeatletter
 \gdef\xxxmark{%
   \expandafter\ifx\csname @mpargs\endcsname\relax 
     \expandafter\ifx\csname @captype\endcsname\relax 
       \marginpar{xxx}
     \else
       xxx 
     \fi
   \else
     xxx 
   \fi}
 \gdef\xxx{\@ifnextchar[\xxx@lab\xxx@nolab}
 \long\gdef\xxx@lab[#1]#2{{\bf [\xxxmark #2 ---{\sc #1}]}}
 \long\gdef\xxx@nolab#1{{\bf [\xxxmark #1]}}
 \gdef\turnoffxxx{\long\gdef\xxx@lab[##1]##2{}\long\gdef\xxx@nolab##1{}}%
}

\title{\vspace{-0.6in}A Simple Yao-Yao-Based Spanner of Bounded Degree
\thanks{This work has been supported by NSF grant CCF-0728909.}}
\author{Mirela Damian\thanks{
Department of Computer Science, Villanova University, Villanova, PA
19085. E-mail: {\tt  mirela.damian@villanova.edu}.}}

\date{}

\begin{document}

\maketitle

\begin{abstract}
\small It is a standing open question to decide whether the \YY\
structure for unit disk graphs (UDGs) is a length spanner of not.
This question is highly relevant to the topology control problem for
wireless ad hoc networks. In this paper we make progress towards
resolving this question by showing that the \YY\ structure is a
length spanner for UDGs of bounded aspect ratio. We also propose a
new local algorithm, called Yao-Sparse-Sink, based on the Yao-Sink
method introduced by Li, Wan, Wang and Frieder, that computes a
$(1+\e)$-spanner of bounded degree for a given UDG and for given $\e
> 0$. The Yao-Sparse-Sink method enables an efficient local
computation of sparse sink trees. Finally, we show that all these
structures for UDGs -- Yao, Yao-Yao, Yao-Sink and Yao-Sparse-Sink --
have arbitrarily large weight.
\end{abstract}

\section{Introduction}
Let $G = (V,E)$ be a connected graph with $n$ vertices embedded in
the Euclidean plane. For any pair of vertices $u, v \in V$, an
$uv$-\emph{path} is defined by a sequence of edges $uu_1, u_1u_2,
\ldots, u_sv$. A subgraph $H$ of $G$ is a \emph{length spanner} of
$G$ if, for all pairs of vertices $u, v \in V$, the length of a
shortest $uv$-path in $H$ is no longer than a constant times the
length of a shortest $uv$-path in $G$; if the constant value is $t$,
$H$ is called a length $t$-spanner and $t$ is called \emph{length
stretch factor}.

The \emph{power} needed to support a wireless link $uv$ is
$|uv|^\beta$, where $\beta$ is a path loss gradient
(a real constant between $2$ and $5$) that depends on the
transmission environment. A subgraph $H$ of a graph $G$ has
\emph{power stretch factor} equal to $\rho$ if, for
all pairs of vertices $u, v$ in $G$, the power of a minimum power
$uv$-path in $H$ is no higher than $\rho$ times the power of a
minimum power $uv$-path in $G$. Li et al.~\cite{li02sparse} showed
that a graph with length stretch factor $\delta$ has power stretch
factor $\delta^\beta$, but the reverse is not necessarily true:


\medskip
\noindent {\bf Fact 1~\cite{li02sparse}.} Any subgraph $H \subseteq
G$ with length stretch factor $\delta$ has
power stretch factor $\delta^\beta$. 

\medskip
\noindent The problem of constructing a sparse spanner of a given
graph has received considerable attention from researchers in
computational geometry and ad-hoc wireless networks; we refer the
reader to the recent book by Narasimhan and Smid~\cite{ns-gsn-07}.
The simplest model of a wireless network graph is the Unit Disk
Graph (UDG): an edge exists in the graph if and only if the
Euclidean distance between its endpoints is no greater than $1$.

It is a standing open question to decide whether the \YY\ structure
for UDGs introduced by Li et al.~\cite{li02sparse} is a spanner of
not. The \YY\ graph (also known as \emph{\Y\ plus reverse \Y}) is
based on the \Y\ graph~\cite{Yao82}, from which a number of edges
are eliminated through a reverse \Y\ process, to ensure bounded
degree at each node.
Progress towards resolving this question has been made by Wang and
Li~\cite{WangLi03}, who showed that the Yao-Yao graph has constant
\emph{power stretch factor} in a \emph{civilized} UDG. For constant
$\lambda
> 0$, a $\lambda$-\emph{civilized} graph is a graph in which no two
nodes are at distance smaller than $\lambda$. Most often wireless
devices in a wireless network can not be too close, so it is
reasonable to model a wireless ad hoc network as a civilized UDG.

In this paper we show that the \YY\ graph for a \emph{civilized} UDG
has constant \emph{length stretch factor} as well. Although several
papers refer to a similar result as appearing in~\cite{li02sparse},
to the best of our knowledge there is no version
of~\cite{li02sparse} that publishes this result. We also analyze the
bounded degree spanner generated by the \emph{Yao-Sink} technique
introduced in~\cite{LiWanWang}. The sink technique replaces each
directed star in the Yao graph consisting of all links directed into
a node $u$, by a tree $T(u)$ with sink $u$ of bounded degree.
We propose an enhanced technique called \emph{Yao-Sparse-Sink} that
filters out some of the edges in the Yao graph prior to applying the
sink technique. This enables an efficient local computation of
sparse sink trees, more appropriate for highly dynamic wireless
network nodes. Our analysis of the Yao-Sparse-Sink method provides
additional insight into the properties of the Yao-Yao structure.
We also show that all these structures for UDGs -- Yao, Yao-Yao,
Yao-Sink and Yao-Sparse-Sink --
have arbitrarily large weight.

The rest of the paper is organized as follows. In
Section~\ref{sec:preliminaries} we introduce some notation and
definitions and discuss previous related work. In
Section~\ref{sec:yys.civilized} we show that the Yao-Yao graph is a
spanner for UDGs of bounded aspect ratio (in particular, for
$\lambda$-civilized UDGs). In Section~\ref{sec:alg} we discuss the
Yao-Sink method and, based on this, we propose a new technique
called Yao-Sparse-Sink that computes sparse sink trees efficiently.
Finally, in Section~\ref{sec:total.weight} we show that all these
structures for UDGs have unbounded weight.

\section{Preliminaries}
\label{sec:preliminaries}
%

\subsection{Definitions and Notation}
\label{sec:notation}
Throughout the paper we use the following notation:
$uv$ denotes the edge with endpoints $u$ and $v$;
$\overrightarrow{uv}$ denotes the edge directed from \emph{source}
node $u$ to \emph{sink} node $v$; $|uv|$ denotes the Euclidean
distance between $u$ and $v$; $p(u \rightsquigarrow v)$ denotes a
simple $uv$-path; and $\oplus$ denotes the concatenation operator.
For any nodes $u$ and $v$, let $K_u$ denote an arbitrary cone with
apex $u$, and $K_u(v)$ denote the cone with apex $u$ containing $v$.
For any edge set $E$ and any cone $K_u$, let $E \cap K_u$ denote the
subset of edges in $E$ incident to $u$ that lie in $K_u$. Similarly,
for a graph $G$ and a cone $K_u$, $G \cap K_u$ is the subset of
edges in $G$ incident to $u$ that lie in $K_u$.
%
The \emph{aspect ratio} of an edge set $E$ is the ratio of the
length of a longest edge in $E$ to the length of a shortest edge in
$E$. The aspect ratio of a graph is defined as the aspect ratio of
its edge set.

We assume that each node $u$ has a unique identifier \id($u$).
Define the identifier $\id(\overrightarrow{uv})$ of a directed edge
$\overrightarrow{uv}$ to be the triplet $(|uv|, \id(u), \id(v))$.
For any pair of directed edges $\overrightarrow{uv}$ and
$\overrightarrow{u'v'}$, we say that $\id(\overrightarrow{uv}) <
\id(\overrightarrow{u'v'})$ if and only if one of the following
conditions holds:
\begin{itemize}
\item [(a)] $|uv| < |u'v'|$
\item [(b)] $|uv| = |u'v'|$ and $\id(u) < \id(u')$
\item [(b)] $|uv| = |u'v'|$ and $\id(u) = \id(u')$ and $\id(v) < \id(v')$
\end{itemize}
For an undirected edge $uv$, define $\id(uv) =
\min\{\id(\overrightarrow{uv}), \id(\overrightarrow{vu})\}$. Note
that according to this definition, each edge has a unique
identifier. This enables us to order any edge set by increasing \id\
of edges.

\subsection{Previous Work}
\label{sec:previous.work}
%
%
Yao~\cite{Yao82} defined the Yao graph $Y_k(G)$ as follows. At each node
$u \in V$, any $k$ equal-separated rays originated at $u$ define $k$ cones;
in each cone, pick the edge $uv$ of smallest \id,
if such an edge exists, and add to the Yao graph the directed edge
$\overrightarrow{uv}$. We call this the \ys, described in
Table~\ref{tab:y}.

\begin{table}[hptb]
\begin{center}
\vspace{1mm} \fbox{
\begin{minipage}[h]{0.8\linewidth}
\centerline{\ys($G=(V,E), k$)}
\vspace{1mm}{\hrule width\linewidth}\vspace{2mm} 
\small{\begin{tabbing}
.....\=......\=.......\=.......\=..................................................\kill
Set $E_Y \leftarrow \phi$ and $Y_k \leftarrow (V,E_Y)$.\\
For each node $u$ \\
   \> Partition the space into $k$ equal-size cones with apex $u$ of angle $\theta = 2\pi/k$\\
   \> (assume that each cone is half-open and half-closed). \\
\\
For each node $u$ and each cone $K_u$ such that $E \cap K_u$ is nonempty \\
 \>  Pick the edge $uv \in E \cap K_u$ with lowest \id($uv$). \\
 \>  Add the directed edge $\overrightarrow{uv}$ to $E_Y$.\\
\\
{\bf Output $Y_k = (V, E_Y)$.} 
\end{tabbing}}
\end{minipage}
}
\vspace{1mm}
\end{center}
\vspace{-1em}\caption{The Yao step.}
\label{tab:y}
\end{table}
It has been shown that the output graph $Y_k$ has maximum node degree
$n-1$ and length stretch factor $\frac{1}{1 - 2 \sin{\pi/k}}$.
The first property (high degree) is the main drawback
of the Yao graph. In wireless networks for example, high degree is
undesirable because nodes communicating with too many nodes directly
may experience large overhead that could otherwise be distributed
among several nodes. The Yao-Yao graph $YY_k$ has been proposed
in~\cite{li02sparse} to overcome this shortcoming: at each node $u$
in the Yao graph, discard all directed edges $\overrightarrow{vu}$
from each cone centered at $u$, except for the one with minimum \id.
This filtering step is described in
Table~\ref{tab:yy}.

\begin{table}[hptb]
\begin{center}
\vspace{1mm} \fbox{
\begin{minipage}[h]{0.8\linewidth}
\centerline{\rys($Y_k = (V,E_Y), k$)}
\vspace{1mm}{\hrule width\linewidth}\vspace{2mm} 
\small{\begin{tabbing}
.....\=......\=.......\=.......\=..................................................\kill
Set $E_{YY} \leftarrow E_{Y}$ and $YY_k \leftarrow (V, E_{YY})$. \\
Use the same cone partition as in the \ys. \\
\\
For each node $v$ and each cone $K_v$  \\
   \> Eliminate from $E_{YY}$ all edges $\overrightarrow{uv}$ with sink $v$ that lie in $K_v$, \\
   \> except for the one with the smallest \id. \\
\\
{\bf Output $YY_k = (V, E_{YY})$ (viewed as an undirected graph)}.
\end{tabbing}}
\end{minipage}
}
\vspace{1mm}
\end{center}
\vspace{-1em}\caption{The reverse Yao step.}
\label{tab:yy}
\end{table}

\noindent The output graph $YY_k$ has maximum node degree $2k$, a
constant. However, the tradeoff is unclear in that the question of
whether $YY_k$ is a spanner or not remains open.

\section{\yy-Spanner for Civilized UDGs}
\label{sec:yys.civilized}
Note that any UDG of constant aspect ratio $\Delta$ is a
$1/\Delta$-civilized UDG, and any $\lambda$-civilized UDG has aspect
ratio $1/\lambda$. Therefore, from here on will refer to
$\lambda$-civilized UDGs only.
The \yys\ algorithm applied on a UDG $G = (V, E)$ comprises the Yao
and reverse-Yao steps:
\begin{itemize}
\item [1.] Execute \ys$(G, k)$. The result is the Yao spanner $Y_k = (V, E_Y)$.
\item [2.] Execute \rys$(Y_k, k)$. The result is the Yao-Yao graph $YY_k = (V, E_{YY})$.
\end{itemize}
We now show that $YY_k$ is a length spanner for any civilized UDG.
In proving this, we will make use of the following lemma:

\begin{lemma}[Czumaj and Zhao~\cite{CZ04}]
Let $0 < \theta < \frac{\pi}{4}$ and $t \ge \frac{1}{\cos\theta -
\sin\theta}$. Let $u, v, z$ be three points in the plane with
$\widehat{vuz} \le \theta$. Suppose further that $|uz| \le |uv|$.
Then the edge $\{u, z\}$ followed by a $t$-spanner path from $z$ to
$v$ is a $t$-spanner path from $u$ to $v$ (see Figure~\ref{fig:cz}).
\label{lem:cz}
\end{lemma}
%
\begin{figure}[htbp]
\centering
\includegraphics[width=0.36\linewidth]{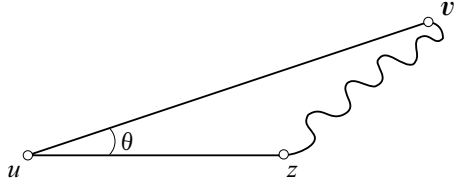}
\caption{If $\theta < \pi/4$ and $p(z \rightsquigarrow v)$ is a
$t$-spanner path, then $uz \oplus p(z \rightsquigarrow v)$ is a
$t$-spanner path.} \label{fig:cz}
\end{figure}

\begin{theorem}
Let $G = (V, E)$ be a $\lambda$-civilized graph, and let $\yy_k$ be
the Yao-Yao structure for $G$. Then $YY_k$ is a spanner with length
stretch factor $t \ge \frac{\lambda}{(\lambda+1)(\cos 2\pi/k-\sin
2\pi/k)-1}$, for any integer $k > 8$ satisfying the condition $(\cos
2\pi/k - \sin 2\pi/k) > \frac{1}{\lambda+1}$.
\label{thm:civilized.spanner}
\end{theorem}
\begin{proof}
The proof is by induction on the rank of edges in the edge set $E$ ordered
by increasing \id. The base case corresponds to the edge $uv \in E$ of
rank 0 (i.e., with smallest $\id(uv)$). Assume without loss of generality
that $\id(uv) = \id(\overrightarrow{uv})$. Since $\overrightarrow{uv}$
has the smallest \id\ among all edges in $K_u(v)$, $\overrightarrow{uv}$
gets added to $E_Y$ in the \ys. Furthermore, since
$\overrightarrow{uv}$ has the smallest \id\ among all edges in
$E_{YY} \cap K_v(u)$ directed into $u$, $\overrightarrow{uv}$ does
not get discarded in the \rys. Thus $uv$ is an edge in $\yy_k$ and so the
theorem holds for the base case.

\begin{figure}[hptb]
\centering
\begin{tabular}{c@{\hspace{0.05\linewidth}}c}
\includegraphics[width=0.43\linewidth]{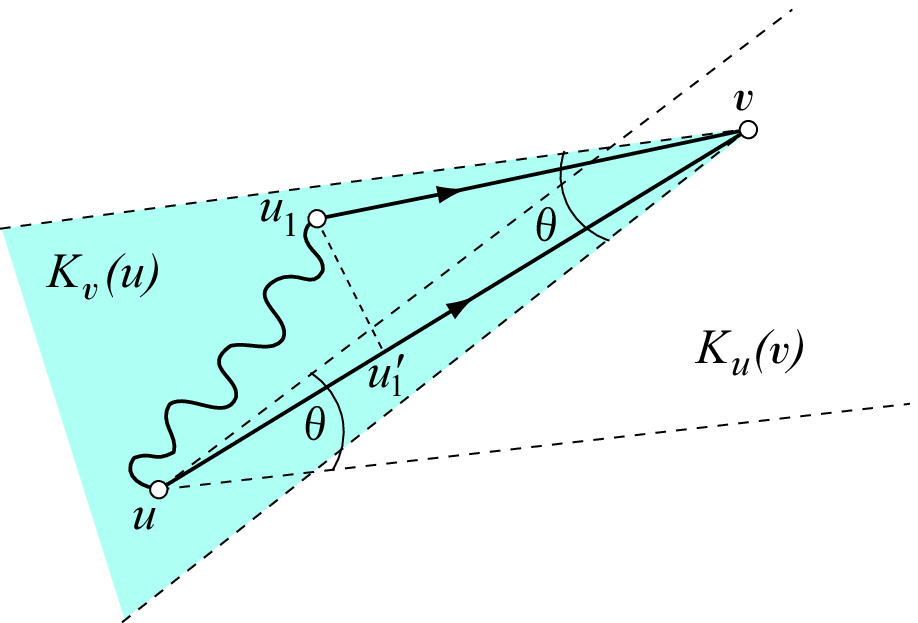} &
\includegraphics[width=0.45\linewidth]{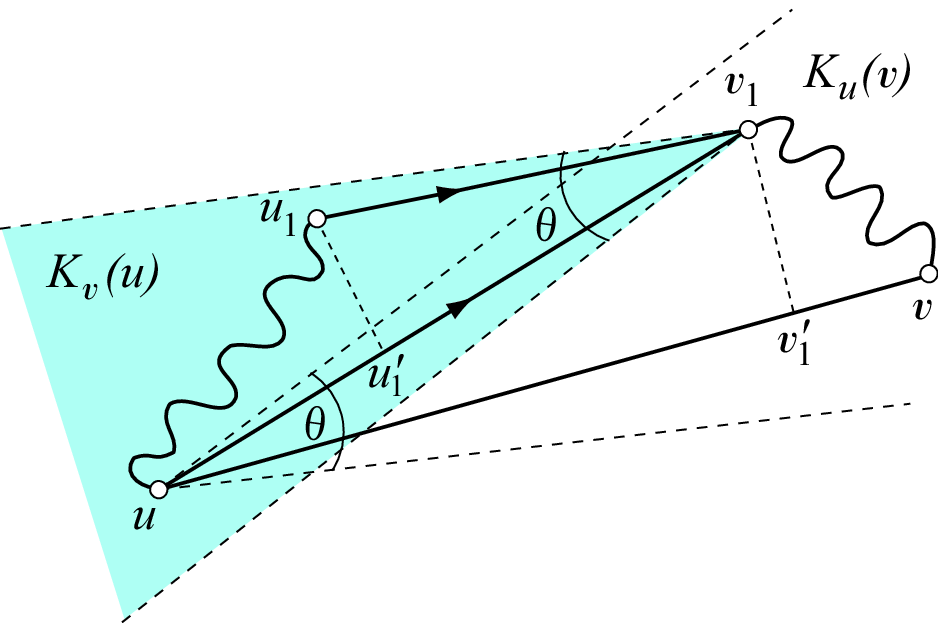} \\
(a) & (b)
\end{tabular}
\caption{Proof of Theorem~\ref{thm:civilized.spanner}: (a) Case 1:
$p(u \rightsquigarrow v) \leftarrow p(u \rightsquigarrow u_1) \oplus
u_1v$; (b) Case 2: $p(u \rightsquigarrow v) \leftarrow p(u
\rightsquigarrow u_1) \oplus u_1v_1 \oplus p(v_1 \rightsquigarrow
v)$.} \label{fig:case12}
\end{figure}

The inductive hypothesis tells that $\yy_k$ contains $t$-spanner
paths between the endpoints of any edge $uv \in E$ of rank no
greater than some value $j \ge 0$. To prove the inductive step,
consider the edge $uv \in E$ of rank $j+1$. Assume without loss of
generality that $\id(uv) = \id(\overrightarrow{uv})$. We discuss two
cases, depending on whether $\overrightarrow{uv}$ belongs to $E_{Y}$
or not. Let $\theta = 2\pi/k < \pi/4$.


\paragraph{Case 1:} $\overrightarrow{uv} \in E_{Y}$. If
$\overrightarrow{uv} \in E_{YY}$ the proof is finished, so assume
the opposite. Note that $\overrightarrow{uv} \not\in E_{YY}$ happens
when $v$ eliminates $\overrightarrow{uv}$ in the \rys\ in favor of
another edge $\overrightarrow{u_1v}$, with
$\id(\overrightarrow{u_1v}) < \id(\overrightarrow{uv})$ (see
Figure~\ref{fig:case12}a). Since $u$ and $u_1$ both lie in a same
cone $K_v$, we have that $\widehat{uvu_1} \le \theta < \pi/4$
and therefore $|u_1 u| < |uv|$. It follows that $\id(u_1u) < \id(uv)$.
Conform the inductive hypothesis,
$YY_k$ contains a $t$-spanner path $p(u \rightsquigarrow u_1)$
from $u$ to $u_1$. These together with Lemma~\ref{lem:cz}
show that $p(u \rightsquigarrow u_1) \oplus u_1v$ is a $t$-spanner
path from $u$ to $v$ in $YY_k$.

\paragraph{Case 2:} $\overrightarrow{uv} \not\in E_{Y}$.
Let $uv_1 \in E \cap K_u(v)$ be the edge selected by
$u$ in the \ys. Thus we have that $\id(uv_1) < \id(uv)$
and therefore $|uv_1| \le |uv|$.
If $\overrightarrow{uv_1} \in E_{YY}$, then arguments
similar to the ones used for Case 1 show that $uv_1 \oplus p(v_1
\rightsquigarrow v)$ is a $t$-spanner path from $u$ to $v$ in $YY_k$; the
existence of a $t$-spanner path $p(v_1 \rightsquigarrow v)$ in $YY_k$
is ensured by the inductive hypothesis.

Consider now the case where
$\overrightarrow{uv_1} \not\in E_{YY}$. 
Since $\overrightarrow{uv_1} \in E_Y$ and $\overrightarrow{uv_1}
\not\in E_{YY}$, the edge $\overrightarrow{uv_1}$ must have been
eliminated by $v_1$ in the \rys\ in favor of
another edge $\overrightarrow{u_1v_1}$, with $\id(u_1v_1) < \id(uv_1)$
(refer to Figure~\ref{fig:case12}b).
Let $u'_1$ be the projection of $u_1$ on $uv_1$.
By the triangle inequality,

\begin{equation}
|uu_1| \le |uu'_1|+|u'_1u_1| = |uv_1| - |u'_1v_1| + |u'_1u_1| \le
|uv_1| - |u_1v_1|\cos\theta + |u_1v_1|\sin\theta. \label{eq:p1}
\end{equation}
Similarly, if $v'_1$ is the projection of $v_1$ on $uv$, we have
\begin{equation}
|v_1v| \le |vv'_1|+|v'_1v_1| = |uv| - |uv'_1| + |v'_1v_1| \le |uv| -
|uv_1|\cos\theta + |uv_1|\sin\theta. \label{eq:p2}
\end{equation}
Since $|uu_1| < |uv_1| \le |uv|$ and $|v_1v| < |uv|$, $\yy_k$
contains $t$-spanner paths $p(u \rightsquigarrow u_1)$ and $p(v_1
\rightsquigarrow v)$ (by the inductive hypothesis).
Let $P_1 = p(u \rightsquigarrow u_1) \oplus 
p(v_1 \rightsquigarrow v)$. We show
that the path $P = P_1 \oplus u_1v_1$ is a $t$-spanner path from $u$ to $v$,
thus proving the inductive step.
The length of $P_1$ is
$$
|P_1| \le t(|uu_1| + |v_1v|).
$$
Substituting inequalities~(\ref{eq:p1}) and~(\ref{eq:p2}) yields
\begin{equation}
|P_1| \le t |uv| + t|uv_1|(1-\cos\theta +\sin\theta) - t|u_1v_1|(\cos\theta  - \sin\theta).
\label{eq:p3}
\end{equation}
Thus the length of $P = P_1 \oplus u_1v_1$ is
\begin{equation}
|P| \le t |uv| + t|uv_1|(1-\cos\theta +\sin\theta) - |u_1v_1|(t\cos\theta  - t\sin\theta -1).
\label{eq:p4}
\end{equation}
Since the input graph $G$ is $\lambda$-civilized, we have that
$|u_1v_1| \ge \lambda$. This along with the inequality~(\ref{eq:p4})
and the fact that $|uv_1| \le 1$ implies
$$
|P| \le t |uv| + (t(1-\cos\theta +\sin\theta) - \lambda(t\cos\theta  - t\sin\theta -1)).
$$
Note that the second term on the right side of the inequality above is
non-positive for any $t \ge \frac{\lambda}{(\lambda+1)(\cos\theta-\sin\theta)-1}$ and for any $\theta$
satisfying the condition $\cos\theta - \sin\theta >
\frac{1}{\lambda+1}$. This completes the proof.
\end{proof}

\medskip
\noindent Theorem~\ref{thm:civilized.spanner} implies that, for
fixed small $\lambda > 0$ and for any $\e > 0$, one can choose
$\theta$ such that $\cos\theta - \sin\theta =
\frac{\lambda+\e+1}{(\lambda+1)(\e+1)} > \frac{1}{\lambda+1}$, to
produce a $t$-spanner $YY_k$ with $t \ge
\frac{\lambda}{(\lambda+1)(\cos\theta-\sin\theta)-1} = 1+\e$. So we
have the following result:

\begin{corollary}
The $YY_k$ structure produced by the \yys\ algorithm for a
given civilized UDG is a spanner with maximum degree $2k$,
length stretch factor $(1+\e)$, and power stretch factor
$(1+\e)^\beta$, for any real $\e>0$ and integer
$k > 8$ satisfying the condition $\cos 2\pi/k - \sin 2\pi/k =
\frac{\lambda+\e+1}{(\lambda+1)(\e+1)}$.
\label{cor:1e}
\end{corollary}

\section{Efficient Local \yssa\ Algorithm for UDGs}
\label{sec:alg}
We have established in Section~\ref{sec:yys.civilized} that the
Yao-Yao graph is a length spanner for civilized UDGs. The question
of whether the Yao-Yao graph is a length spanner for arbitrary UDGs
remains open. 
In order to guarantee both bounded degree and the length spanner
property, Li et al.~\cite{LiWanWang} suggest a sparse topology,
called \emph{Yao-Sink}. Let $G= (V, E)$ be a UDG. The Yao-Sink
algorithm applied on $G$ consists of two steps: (1) Execute
\ys($G,k$) to produce the Yao spanner $Y_k = (V, E_k)$, and (2)
Execute \yss($Y_k, k$) to reduce the degree of $Y_k$. The \yss\ is
described in detail in Table~\ref{tab:ys}.
In ~\cite{YaoLi05} the authors show that the output $YS_k$ generated
by the Yao-Sink method has maximum degree $k(k+2)$ and length
stretch factor $\left(\frac{1}{1 - 2 \sin{\pi/k}}\right)^2$. In
fact, the authors show a more general result that applies to
\emph{mutual inclusion graphs}, which allow for non-uniform
transmission ranges at nodes.

\begin{table}[hptb]
\begin{center}
\vspace{1mm} \fbox{
\begin{minipage}[h]{0.8\linewidth}
\centerline{\yss($Y_k = (V,E_Y), k$)}
\vspace{1mm}{\hrule width\linewidth}\vspace{2mm} 
\small{\begin{tabbing}
.....\= .......\=..........\=.......\=.......\=..................................................\kill
Use the same cone partition as in the \ys. \\
1. \> Set $E_{YS} \leftarrow \emptyset$ and $YS_k \leftarrow (V, E_{YS})$. \\
2. \> For each node $v$ and each cone $K_v$  \\
\vspace{2mm}
\>   {\bf \{Build the tree $T(v)$ corresponding to $K_v$.\} } \\
\>   2.1 \> Let $I$ be the set of vertices $u$ such that $\overrightarrow{uv} \in E_{Y} \cap K_v$. \\
\>       \> Set $I(v) \leftarrow I$. Initialize the ordered vertex sequence $J \leftarrow (v)$. \\
\>   2.2 \> Initialize $T(v) \leftarrow \emptyset$. Repeat until $I$ is empty \\
\>   \> 2.2.1 \> Remove the first vertex $u$ from the sequence $J$. \\
\>   \> 2.2.2 \> For each cone $K_u$ \\
\>   \> \> \> Let $w \in I(u) \cap K_u$ be the node that minimizes $\id(\overrightarrow{wu})$ (if any). \\
\>   \> \> \> Add $\overrightarrow{wu}$ to $T(v)$ and move $w$ from $I$ to $J$. \\
\>   \> \> \> Set $I(w) \leftarrow I(u) \cap K_u$. \\
\>  2.3  \> Add all edges of $T(v)$ to $E_{YS}$. \\
\\
{\bf Output $YS_k = (V, E_{YS})$ (viewed as an undirected graph)}.
\end{tabbing}}
\end{minipage}
}
\vspace{1mm}
\end{center}
\vspace{-1em}\caption{The Sink step.}
\label{tab:ys}
\end{table}

The following two lemmas (Lemmas~\ref{lem:ys.seq}
and~\ref{lem:ys.seq2}) identify two important properties of the
output spanner $YS_k$ generated by the Yao-Sink method.
Specifically, they show the existence of a particular path in $YS_k$
corresponding to each Yao edge removed in the \yss.


\begin{lemma}
For each edge $\overrightarrow{uv} \in E_Y$, there is a $uv$-path
$\Pi = w_0w_1, w_1w_2, \ldots, w_{h-1}w_h$ in $K_v(u)$, with $w_0 =
v$ and $w_h = u$,
such that $\overrightarrow{w_{i}w_{i-1}} \in E_{YS} \cap
K_{w_{i-1}}(u)$ and $\id(\overrightarrow{w_{i}w_{i-1}}) <
\id(\overrightarrow{uw_{i-1}})$, for each $i = 1, 2, \ldots, h$.
\label{lem:ys.seq}
\end{lemma}
\begin{proof}
Let $I = I(v)$ be the vertex set defined in Step 2.1 of \yss\ for
node $v$ and cone $K_v(u)$. If $u \in I(v) \cap K_v(u)$ minimizes
$\id(\overrightarrow{uv})$, then the path sought is $\Pi = vu$ and
the proof is finished. Otherwise, let $\Pi = w_0w_1, w_1w_2, \ldots,
w_{p-1}w_p$ be a longest path in $K_v(u)$ that satisfies the
conditions of the lemma: $\overrightarrow{w_{i}w_{i-1}} \in E_{YS}
\cap K_{w_{i-1}}(u)$ and $\id(\overrightarrow{w_{i}w_{i-1}}) <
\id(\overrightarrow{uw_{i-1}})$, for each $i = 1, 2, \ldots, p$. We
prove by contradiction that $w_p = u$. Assume to the contrary that
$w_p$ and $u$ are distinct.
\begin{figure}[htbp]
\centering
\begin{tabular}{c@{\hspace{0.06\linewidth}}c}
\includegraphics[width=0.45\linewidth]{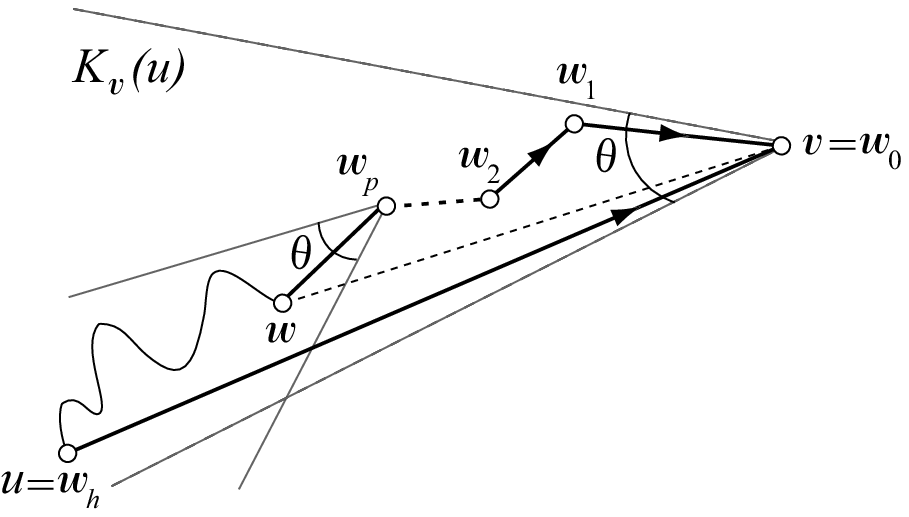} &
\includegraphics[width=0.45\linewidth]{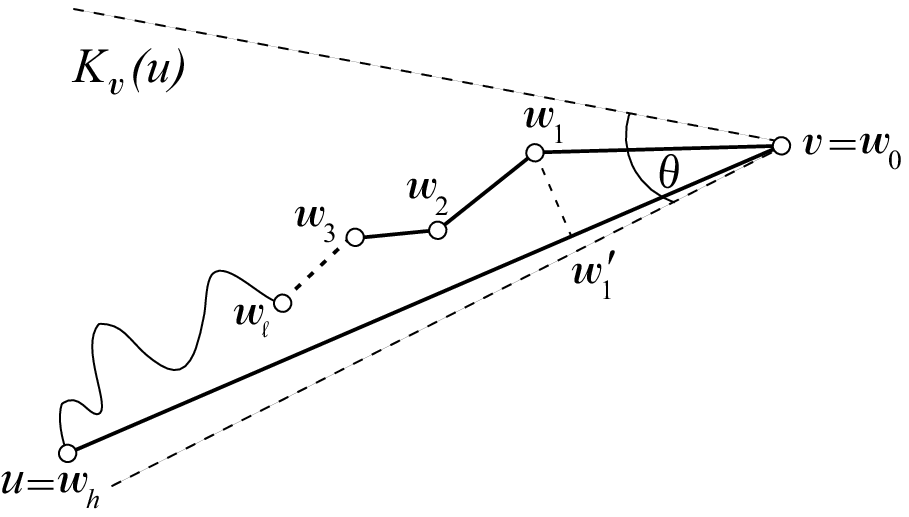} \\
(a) & (b)
\end{tabular}
\caption{(a) Lemma~\ref{lem:ys.seq}: path $\Pi = w_0w_1, w_1w_2,
\ldots, w_{h-1}w_h$ (b) Lemma~\ref{lem:ys.seq2}: $\widehat{w_{1}uv}
\le \theta$. } \label{fig:alg}
\end{figure}
%
Since $\overrightarrow{w_{p}w_{p-1}} \in E_{YS}$, it must be that
$w_p \in I(w_{p-1})$. Furthermore, since $w_p$ and $u$ lie in a same
cone $K_{w_{p-1}}(u)$, the set $I(w_p)$ defined in Step 2.2.2 of the
\yss\ is $I(w_p) = I(w_{p-1}) \cap K_{w_{p-1}}(u)$ and includes both
$w_p$ and $u$.

Consider now the instance when $w_p$ gets processed (i.e, it gets
removed from $J$ in Step 2.2.1 of \yss). See Figure~\ref{fig:alg}a.
%
%
First observe that $I(w_p) \cap K_{w_{p}}(u)$ is nonempty, since it
contains at least the node $u$. This implies that there exists $w
\in I(w_p) \cap K_{w_p}(u)$ that minimizes
$\id(\overrightarrow{ww_p})$. It follows that
$\overrightarrow{ww_{p}} \in E_{YS} \cap K_{w_{p}}(u)$ and either $w
= u$ or $\id(\overrightarrow{ww_{p}}) <
\id(\overrightarrow{uw_{p}})$. Either case contradicts our
assumption that $\Pi$ is a longest path that satisfies the
conditions of the lemma.
\end{proof}

\medskip
\noindent In the context of Lemma~\ref{lem:ys.seq}, we next prove
the existence of a long enough subpath of $\Pi$ from $v$ to one of
the vertices $w_\ell \in \Pi$ that closely approximates the direct
link $vw_\ell$.

\begin{lemma}
Let $\Pi = w_0w_1, w_1w_2, \ldots, w_{h-1}w_h$, with $w_0 = v$ and
$w_h = u$,  be the path identified in Lemma~\ref{lem:ys.seq}
corresponding to a given edge $\overrightarrow{uv} \in E_Y$. Then
there exists $\ell \le h$ such that $|vw_{\ell}| \ge
|uv|/(2\cos\theta)$ and
\begin{eqnarray} \nonumber \sum_{i=0}^{\ell-1} |w_iw_{i+1}| \le
\frac{|vw_\ell|}{\cos 2\theta}.
\end{eqnarray}
\label{lem:ys.seq2}
\end{lemma}
\begin{proof}
Let $\ell \le h$ be the smallest index in the sequence $1, 2,
\ldots, h$ such that $|vw_{\ell}| \ge |uv|/(2\cos\theta)$. Since
$|vw_{h}| = |uv| > |uv|/(2\cos\theta)$, such an index always exists.
Let $w'_{i}$ be the projection of $w_{i}$ on $uv$, for each $i$. We
first prove that the following invariant holds:
\begin{itemize}
\item[(a)] $\widehat{w_ivu} \le \theta$, for each $i = 0, 1, \ldots, \ell$.
\item[(b)] $\widehat{w_iuv} \le \theta$, for each $i = 1, \ldots, \ell-1$.
\item[(c)] $|w_iw_{i-1}| \le |w'_iw'_{i-1}|/\cos 2\theta$, for each $i = 1, \ldots, \ell$.
\end{itemize}
Property (a) follows immediately from the fact that $w_i$ and $u$
belong to a same cone $K_v(u)$, for each $i$. The proof for
properties (b) and (c) is by induction on the index $i$. The base
case corresponds to $i = 1$ (i.e, $\Pi = w_0w_1$). See
Figure~\ref{fig:alg}b. We prove that $\widehat{w_1uv} \le \theta$
(claim (b) for the case when $\ell \ge 2$). First observe that
$|w_{1} v| < |uv| /(2\cos \theta)$, otherwise it would contradict
our choice of $\ell$. Thus we have that
$$\tan\widehat{w_{1}uv} = \frac{|w_{1}w'_{1}|}{|uv|-|w'_{1}v|}
                        = \frac{|w_{1}v| \sin\widehat{w_{1}vu}}{|uv|-|w_{1}v|\cos\widehat{w_{1}vu}}
                        \le \frac{|w_{1}v| \sin\theta}{|uv|-|w_{1}v|\cos\theta} < \tan\theta.$$
If follows that $\widehat{w_{1}uv} < \theta$, so claim (b) holds.
For claim (c), note that $|vw_{1}| \le |vw'_{1}|/\cos\theta <
|vw'_{1}|/\cos 2\theta$.
Assume that the claim holds for any index less than $i$ , for some
$i > 1$. To prove the inductive step, consider a path $\Pi = w_0w_1,
\ldots, w_{i-1}w_i$, with $i \le \ell$.
\begin{figure}[hptb]
\centering
\includegraphics[width=0.98\linewidth]{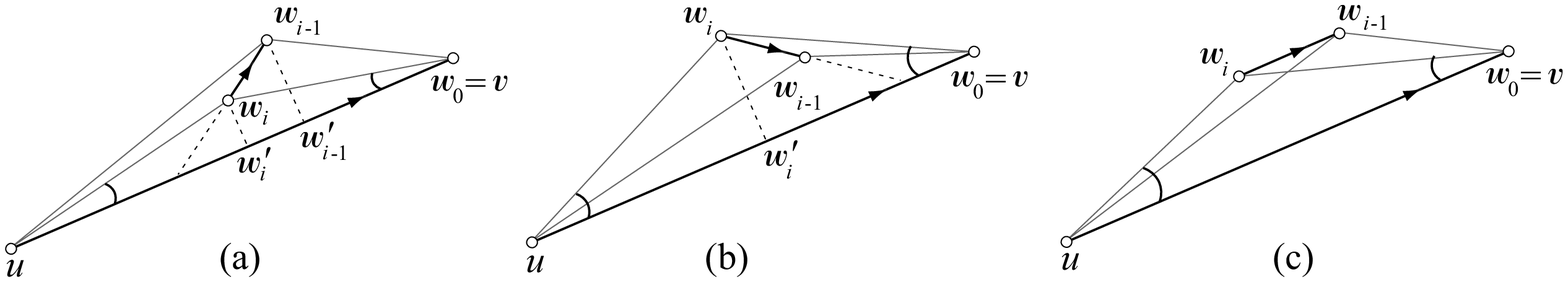}
\caption{Lemma~\ref{lem:ys.seq2} proof: (a) $w_{i} \in \triangle
uw_{i-1}v$ (b) $w_{i-1} \in \triangle uw_{i}v$ (c) Neither (a) not
(b) holds. } \label{fig:seq}
\end{figure}
We distinguish three cases:
\begin{itemize}
\item[(i)] $w_{i} \in \triangle uw_{i-1}v$ (see
Figure~\ref{fig:seq}a). Then $\widehat{w_{i}vu} <
\widehat{w_{i-1}vu} \le \theta$ (this latter inequality is true by
the inductive hypothesis). Also note that the angle formed by
$w_{i-1}w_i$ and $uv$ is no greater than
$\widehat{uw_{i-1}w_i}+\widehat{w_{i-1}uv} \le 2\theta$ and
therefore $|w_iw_{i-1}| \le |w'_iw'_{i-1}|/\cos 2\theta$.

\item[(ii)] $w_{i-1} \in \triangle uw_{i}v$ (see
Figure~\ref{fig:seq}b). 
%
We show that $\widehat{w_{i}uv} \le \theta$ (claim (b) for the case
when $i \le \ell-1$). First observe that the condition $|w_i v| <
|uv| /(2\cos \theta)$ must hold for each $i \le \ell-1$; otherwise,
we could find a lower index $i < \ell$ satisfying the condition
$|vw_i| \ge |uv| /(2\cos \theta)$, contradicting our choice of
$\ell$. 
As before, we have that
$$\tan\widehat{w_{i}uv} = \frac{|w_{i}w'_{i}|}{|uv|-|w'_{i}v|}
                        = \frac{|w_{i}v| \sin\widehat{w_{i}vu}}{|uv|-|w_{i}v|\cos\widehat{w_{i}vu}}
                        \le \frac{|w_{i}v| \sin\theta}{|uv|-|w_{i}v|\cos\theta} < \tan\theta.$$
If follows that $\widehat{w_{i}uv} < \theta$. Also note that in this
case the angle formed by $w_{i-1}w_i$ and $uv$ is no greater than
$\widehat{w_{i}w_{i-1}u} < \theta$ and therefore $|w_iw_{i-1}| \le
|w'_iw'_{i-1}|/\cos\theta < |w'_iw'_{i-1}|/\cos 2\theta$.

\item[(iii)] Neither (i) not (ii) holds (see Figure~\ref{fig:seq}c).
Arguments identical to the ones used in case (ii) above show that
$\widehat{w_{i}uv} \le \theta$. Also note that the angle formed by
$w_{i-1}w_i$ and $uv$ is no greater than $\max\{\widehat{w_{i-1}uv},
\widehat{w_{i}vu}\} \le \theta$ and therefore $|w_iw_{i-1}| \le
|w'_iw'_{i-1}|/\cos\theta < |w'_iw'_{i-1}|/\cos 2\theta$.
\end{itemize}
We have shown that $|w_iw_{i-1}| \le |w'_iw'_{i-1}|/\cos 2\theta$
for each $i = 1, 2, \ldots, \ell$. Summing up over $i$ yields
\begin{eqnarray} \nonumber \sum_{i=0}^{\ell-1} |w_iw_{i+1}| \le
\frac{|w'_{\ell}v|}{\cos 2\theta} < \frac{|w_{ \ell}v|}{\cos
2\theta}.
\end{eqnarray}
This completes the proof.
\end{proof}

\medskip
We will show that Lemmas~\ref{lem:ys.seq} and~\ref{lem:ys.seq2}
enable us to discard some Yao edges from $Y_k$ prior to executing
the \yss, without compromising the spanner property. This leads to
the construction of efficient sparse sink trees in the \yss.
Table~\ref{tab:yssa} describes our method called \yssa\ that
incorporates this intermediate edge filtering step. In the filtering
step, each node $u$ partitions the set of Yao edges incident to $u$
into a number of subsets, such that all edges in a same subset $F_i$
have similar sizes. The aspect ratio of each subset $F_i$ is
controlled by the input parameter $r > 1$. From each subset $F_i$,
only the Yao edge with smallest \id\ is carried on to the \yss; all
other Yao edges from $F_i$ are discarded.
\begin{table}
\begin{center}
\vspace{1mm} \fbox{
\begin{minipage}[h]{0.8\linewidth}
\centerline{Algorithm \yssa($G = (V,E), k, r$)}
\vspace{1mm}{\hrule width\linewidth}\vspace{2mm} 
\small{\begin{tabbing} .....\=
.......\=........\=.......\=.......\=..................................................\kill
1. \> Execute \ys($G, k$). The result is the Yao spanner $Y_k = (V, E_Y)$. \\
\\
2. \> For each node $v \in V$ and each cone $K_v$ \\
   \> \> Let $F \subseteq E_Y \cap K_v$ be the subset of Yao edges from $K_v$ directed into
   $v$ \\
   \> \> Let $\overrightarrow{uv} \in F$ be the edge of minimum \id.
        Let $\Delta$ be the aspect ratio of $F$. \\
   \> \> Partition $F$ into disjoint subsets $F_1, F_2, \ldots, F_s$, with $s = \lceil\log_r\Delta\rceil$, such that \\
   \> \> \> $F_i = \{ab \in F ~|~ |uv|r^{i-1} \le |ab| < |uv|r^{i} \}$. \\
   \> \>  For each $i = 1, 2, \ldots, s$ \\
   \> \> \>  Add to $E_{YE}$ (initially $\emptyset$) the edge from $F_i$ of smallest \id. \\
   \> {\bf Result is $YE_k = (V, E_{YE})$ of degree $O(\log\Delta)$.} \\
\\
3. \> Execute \yss($YE_k, k$). The result is $YES_k = (V, E_{YES})$. \\
\\
{\bf Output $YES_k$ (viewed as an undirected graph)}.
\end{tabbing}}
\end{minipage}
}
\vspace{1mm}
\end{center}
\vspace{-1em}\caption{The Yao-Sparse-Sink algorithm.}
\label{tab:yssa}
\end{table}

It can be verified that the result $YE_k$ of the filtering step is a
spanner for $G$ of maximum degree $O(\log \Delta)$, where $\Delta$
is the aspect ratio of $G$. Because of space constraints we skip
this proof and turn instead to showing that the output $YES_k$
generated by the \yssa\ method is a spanner of constant maximum
degree. Intuitively, if $YES_k$ contains short paths between the
endpoints of an edge processed in the \yss, then $YES_k$ contains
short paths between the endpoints of all nearby edges of similar
sizes.

\begin{theorem}
Let $G = (V, E)$ be a UDG and let $r > 1$, $k \ge 8$, $\theta =
2\pi/k$ and $\lambda = \frac{1}{2r\cos\theta}$  be constants such
that $(\cos \theta - \sin \theta) > \frac{\lambda}{\lambda+1}$. When
run with these values of $r$ and $k$, the output of the \yssa\
algorithm is a $t$-spanner of degree $k(k+2)$, for any $t \ge
\frac{\lambda/\cos(2\theta)}{(\lambda+1)(\cos \theta-\sin
\theta)-1}$. \label{thm:spanner}
\end{theorem}
\begin{proof}
The degree of $YES_k$ is no greater than the degree of the Yao-Sink
spanner, which is $k(k+2)$~\cite{YaoLi05}. We now prove that $YES_k$
is a $t$-spanner. The proof is by induction on the rank of edges in
the set $E$ ordered by increasing \id. The base case corresponds to
the edge $uv \in E$ of minimum \id. Arguments similar to the ones
used for the base case in Theorem~\ref{thm:civilized.spanner} show
that $uv \in YYS_k$.

The inductive hypothesis tells that $YES_k$ contains $t$-spanner
paths between the endpoints of any edge $uv \in E$ whose rank is no
greater than some value $j \ge 0$. To prove the inductive step,
consider the edge $uv \in E$ of rank $j+1$. Assume without loss of
generality that $\id(uv) = \id(\overrightarrow{uv})$. We discuss two
cases, depending on whether $\overrightarrow{uv}$ belongs to $E_{Y}$
or not.

\paragraph{Case 1:} $\overrightarrow{uv} \in E_{Y}$. Assume first
that $\overrightarrow{uv} \in E_{YE}$. By Lemma~\ref{lem:ys.seq2},
$YES_k$ contains an edge $\overrightarrow{w_1v} \in K_v(u)$. This
implies that $\widehat{uvw_1} \le \theta < \pi/4$. Furthermore,
since $\id(w_1v) < \id(uv)$ (and therefore $|w_1v| \le |uv|$), we
have that $|uw_1| < |uv|$ (see Figure~\ref{fig:algcases}a). Thus we
can use the inductive hypothesis to show that $YES_k$ contains a
$t$-spanner path $p(u \rightsquigarrow w_1)$. By Lemma~\ref{lem:cz},
$p(u \rightsquigarrow w_1) \oplus w_1v$ is a $t$-spanner path in
$YES_k$ from $u$ to $v$.
%
\begin{figure}[htbp]
\centering
\includegraphics[width=0.98\linewidth]{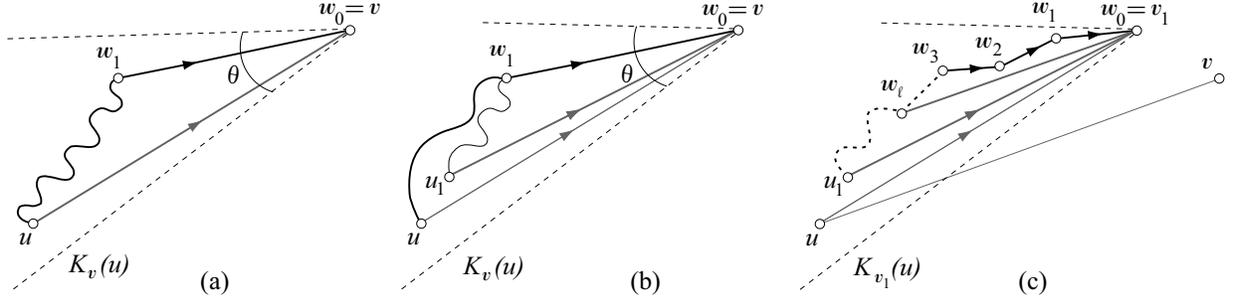}
\caption{Proof of Theorem~\ref{thm:spanner}: (a) Case $\vec{uv} \in
E_{YE}$ (b) Case $\vec{uv} \in E_Y \setminus E_{YE}$ (c) Case
$\vec{uv} \not\in E_Y$.} \label{fig:algcases}
\end{figure}

Assume now that $\overrightarrow{uv} \not \in E_{YE}$. Let $F$ be
the edge set corresponding to cone $K_v(u)$ and let $i$ be such that
$\overrightarrow{uv} \in F_i$. Since $uv \not \in E_{YE}$ there is
an edge $\overrightarrow{u_1v} \in F_i$ of smaller \id\ that gets
added to $E_{YE}$. Note however that $u_1$ and $u$ belong to one
same cone $K_v(u)$ (see Figure~\ref{fig:algcases}b). By the same
arguments as above, there is $\overrightarrow{w_1v} \in K_v(u)$
corresponding to the edge $\overrightarrow{u_1v} \in E_{YE}$ which
enables us to identify the $t$-spanner path $p(u \rightsquigarrow
w_1) \oplus w_1v$ from $u$ to $v$ in $YES_k$.

\paragraph{Case 2:} $\overrightarrow{uv} \not\in E_{Y}$.
Let $uv_1 \in E_Y$ be the edge selected by $u$ in the \ys. Thus we
have that $\id(uv_1) < \id(uv)$ and therefore $|uv_1| \le |uv|$. If
$\overrightarrow{uv_1} \in E_{YES}$, then arguments similar to the
ones used for Case 1 show that $uv_1 \oplus p(v_1 \rightsquigarrow
v)$ is a $t$-spanner path from $u$ to $v$ in $YES_k$.

Consider now the case when $\overrightarrow{uv_1} \not\in E_{YES}$.
Let $F$ be the edge set corresponding to cone $K_{v_1}(u)$ and let
$i$ be such that $\overrightarrow{uv_1} \in F_i$. Let
$\overrightarrow{u_1v_1} \in F_i$ be the edge that gets added to
$E_{YE}$ in Step 2 of the \yssa\ algorithm. Since both $uv_1$ and
$u_1v_1$ belong to a same set $F_i$, and since
$\id(\overrightarrow{u_1v_1}) \le \id(\overrightarrow{uv_1})$
(equality happens when $u_1 = u$), we have that
\begin{equation}
|u_1v_1| \ge |uv_1|/r. 
\label{eq:alg1}
\end{equation}
Lemma~\ref{lem:ys.seq2} indicates that, corresponding to the edge
$\overrightarrow{u_1v_1} \in E_{YE}$, there exists a path $P_0 \in
YES_k \cap K_{v_1}(u_1)$ extending from $w_0 = v_1$ to some vertex
$w_{\ell}$, such that 
\begin{eqnarray}
|P_0| & \le & |w_\ell v_1|/\cos 2\theta. \label{eq:alg2} \\
|w_\ell v_1| & \ge & |u_1v_1|/(2\cos\theta) \label{eq:alg20}
\end{eqnarray}

Now note that, since $|v_1w_{\ell}| \le |v_1 u_1|$ and since
$\widehat{w_{\ell}v_1u} \le \theta \le \pi/4$, we have that
$|uw_\ell| < |uv_1| \le |uv|$. Similarly, $|v_1v| \le |uv|$. Thus we
can use the inductive hypothesis to claim the existence of
$t$-spanner paths  $p(u \rightsquigarrow w_{\ell})$ and $p(v_1
\rightsquigarrow v)$. Let $P_1 = p(u \rightsquigarrow w_{\ell})
\oplus p(v_1 \rightsquigarrow v)$. We show that $P = P_0 \oplus P_1$
is a $t$-spanner path from $u$ to $v$. Calculations identical to the
ones used to derive the inequality~(\ref{eq:p3}) yield
\begin{equation}
\nonumber |P_1| \le t |uv| + t|uv_1|(1-\cos\theta +\sin\theta) -
t|w_{\ell}v_1|(\cos\theta  - \sin\theta). \label{eq:alg3}
\end{equation}
This along with~(\ref{eq:alg2}) shows that the length of $P = P_0
\oplus P_1$ is
\begin{equation}
\nonumber |P| \le t |uv| + t|uv_1|(1-\cos\theta +\sin\theta) -
|w_{\ell}v_1|(t\cos\theta  - t\sin\theta -1/\cos 2\theta).
\label{eq:alg4}
\end{equation}
Substituting~(\ref{eq:alg1}) and~(\ref{eq:alg2}) yields
\begin{equation}
|P| \le t |uv| + (t|uv_1|(1-\cos\theta +\sin\theta) -
\frac{|uv_1|}{2r\cos\theta}(t\cos\theta  - t\sin\theta -1/\cos
2\theta)). \label{eq:alg5}
\end{equation}
Note that the second term in the right hand side of the
inequality~(\ref{eq:alg5}) is non-positive for any $t \ge
\frac{\lambda/cos 2\theta}{(\lambda+1)(\cos\theta-\sin\theta)-1}$
and for any $\theta$ satisfying the condition $\cos\theta -
\sin\theta > \frac{1}{\lambda+1}$. 
\end{proof}

\medskip
\noindent Arguments similar to the ones used for
Corollary~\ref{cor:1e} show that, for appropriate values of $r$ and
$k$ corresponding to a fixed $\e > 0$, $YYS_k$ is a
$(1+\e)$-spanner.

\paragraph{Efficient Local Implementation.}
For a local implementation of the \ys, the authors propose
in~\cite{li02sparse} to have each sink node $u$ build $T(u)$ and
then broadcast $T(u)$ to all nodes in $T(u)$. It can be easily
verified that, for each node $u$ and each cone $K_u$, the neighbors
of $u$ that lie in $K_u$ (including $u$) form a clique. This
suggests a more efficient alternate local implementation of the \ys:
each node collects the coordinate information from its immediate
neighbors, then simulates the execution of the \ys\ locally, on the
collected neighborhood. This implementation avoids broadcasting
messages of size $O(n)$ (encoding the sink trees) by each node, thus
saving some battery power. This idea can be extended to the \yssa\
algorithm as well: each node collects its neighborhood information
in one round of communication, then simulates the execution the
\yssa\ algorithm on the collected neighborhood.

\section{Total Weight of $Y_k$, $YY_k$, $YS_k$ and $YES_k$}
\label{sec:total.weight}

Define the total \emph{weight} $wt(G)$ of a graph $G$ as the sum of
the lengths of its constituent edges. We first show that the total
weight of the Yao graph $Y_k$ constructed by \ys\ is arbitrarily
high compared to the weight of the Minimum Spanning Tree (MST) for
$V$. Although this result is fairly straightforward, to the best of
our knowledge it has not appeared in the literature.

\begin{theorem}
Let $G$ be a UDG and let $Y_k =$ \ys($G, k$). Then $wt(Y_k) =
\Omega(n)\cdot wt(MST)$.
\end{theorem}
\begin{proof}
Consider a set of $n=2s$ nodes equally distributed along the top and
bottom sides of a unit square, as in Figure~\ref{fig:weight}. Let
$u_1, u_2, \ldots u_s$ denote the top nodes and $v_1, v_2, \ldots
v_s$ the bottom nodes.
%
\begin{figure}[htbp]
\centering
\begin{tabular}{c@{\hspace{0.15\linewidth}}c}
\includegraphics[width=0.25\linewidth]{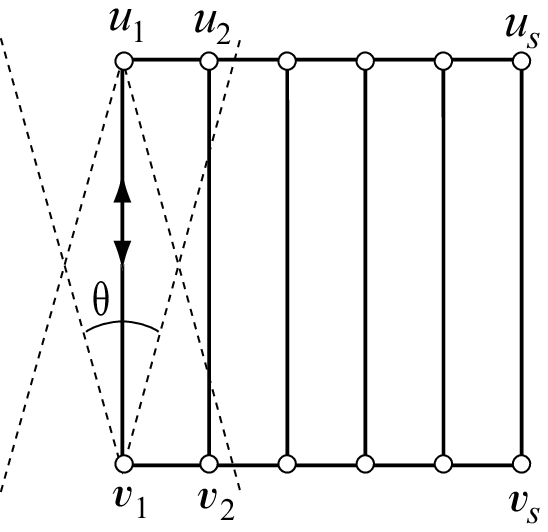} &
\includegraphics[width=0.25\linewidth]{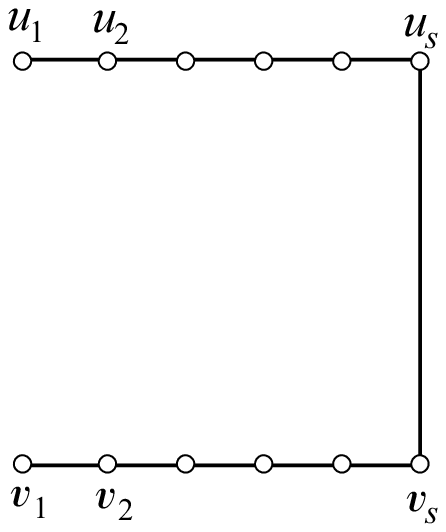} \\
(a) & (b)
\end{tabular}
\caption{$Y_k$ has unbounded weight: (a) $wt(Y_k) = n/2+2$; (b)
$wt(MST) = 3$.} \label{fig:weight}
\end{figure}
%
Observe that for each node $u_i$, the angular distance between any
of its left/right neighbors and $v_i$ is $\pi/2$. This means that
the only edge incident to $u_i$ that lies in the cone of angle
$2\pi/k \le \pi/3$ (for $k \ge 6$) centered at $u_i$ and containing
$v_i$ is $u_iv_i$ (see Figure~\ref{fig:weight}a). Consequently,
$u_i$ adds $\overrightarrow{u_iv_i}$ to $E_Y$ in the \ys. Similarly,
$v_i$ adds $\overrightarrow{v_iu_i}$ to $E_Y$. Thus the total weight
of $Y_k$ is no less than
$$
  \sum_{i=1}^s |u_iv_i| = n/2.
$$
However, the weight of the spanning tree illustrated in
Figure~\ref{fig:weight}b is 3. This completes the proof.
\end{proof}

\medskip
\noindent We have shown that, for any node $u$ in the topology from
Figure~\ref{fig:weight}a, at most one edge from $Y_k$ lies in any
cone $K_u$ centered at $u$. This implies that:

\begin{itemize}
\item [(a)] No edges from $Y_k$ get discarded
in the \rys.
\item [(a)] No edges from $Y_k$ get discarded
in the filtering step (Step 2) of \yssa.
\item [(b)] No edges from $Y_k$ get altered in the \yss.
\end{itemize}
This implies that $YY_k$, $YS_k$ and $YES_k$ are all identical to
$Y_k$ and therefore have unbounded weight as well. 
%

It is worth noting that any civilized UDG $G=(V, E)$ has weight
within a constant factor of $wt(MST(V))$ and therefore the
structures $Y_k$, $YY_k$, $YS_k$ and $YES_k$ for civilized UDGs have
bounded weight as well. This follows immediately from a result
obtained by Das et al.~\cite{Das95}:

\medskip
\noindent {\bf Fact 2 (Theorem 1.2 in~\cite{Das95}).} If a set of
line segments $E$ satisfies the isolation property, then $wt(E) =
O(1)\cdot wt(MST)$.

\medskip
\noindent A set of line segments $E$ is said to satisfy the
\emph{isolation property} if each segment $uv \in E$ can be
associated with a cylinder $B$ of height and width equal to $c|uv|$,
for some constant $c > 0$, such that the axis of $B$ is a subsegment
of $uv$ and $B$ does not intersect any line segment other than $uv$.
In the case of $\lambda$-civilized UDGs, this property is satisfied
by $c = \lambda$.

\section{Conclusions}
We have shown that the Yao-Yao graph is a spanner for UDGs of
bounded aspect ratio. We have also proposed an extension of the
Yao-Sink method, called Yao-Sparse-Sink, that enables an efficient
local computation of sparse sink trees. The Yao-Sparse-Sink method
is preferable to the Yao-Sink method for topology control in highly
dynamic wireless environments. Our analysis of the Yao-Sparse-Sink
method provides additional insight into the properties of the
Yao-Yao structure. However, the main question of whether the Yao-Yao
graph for arbitrary UDGs is a length spanner or not remains open.


\medskip
\noindent {\bf Acknowledgement.} We thank Michiel Smid for helpful
discussions on these problems.

\small
\def\cprime{$'$}

\end{document}